\documentclass[11pt]{article}
\usepackage{graphicx,verbatim,hyperref,amssymb,amsmath,amsthm,amsfonts,enumerate,multicol, xspace,comment,xcolor, algorithmic, algorithm, cite,fullpage, enumitem, bbold, authblk, thm-restate, thmtools, appendix, array}
\usepackage[linesnumbered,boxed, algo2e]{algorithm2e, empheq}

\newtheorem{theorem}{Theorem}
\newtheorem{corollary}{Corollary}[theorem]
\newtheorem{lemma}{Lemma}[section]
\newtheorem{claim}[lemma]{Claim}

\newtheorem{definition}{Definition}

\newcommand{\tsp}{\mathrm{TSP}}
\newcommand{\mtsp}{\mathrm{mTSP}}
\newcommand{\mvTSP}{\mathrm{MV\mbox{-}TSP}}
\newcommand{\MVmTSP}{\mathrm{MV\mbox{-}mTSP}}
\newcommand{\SDTSP}{\mathrm{SD\mbox{-}mTSP}}
\newcommand{\SDMVTSP}{\mathrm{SD\mbox{-}MV\mbox{-}mTSP}}
\usepackage[textsize=tiny,textwidth=2cm,color=green!50!gray]{todonotes} 

\SetKwInput{KwInput}{Input}                
\SetKwInput{KwOutput}{Output}              

\title{Linear Programming based Reductions for Multiple Visit TSP and Vehicle Routing Problems}

\begin{document}

\author[1]{Aditya Pillai\thanks{apillai32@gatech.edu}}
\author[1]{Mohit Singh\thanks{mohit.singh@isye.gatech.edu}}
\affil[1]{H. Milton Stewart School of Industrial and Systems Engineering, Georgia Institute of Technology, Atlanta, USA}

\maketitle

\begin{abstract}
    Multiple TSP ($\mtsp$) is a important variant of $\tsp$ where a set of $k$ salesperson together visit a set of $n$ cities. The $\mtsp$ problem has applications to many real life applications such as vehicle routing. Rothkopf \cite{Rothkopf66} introduced another variant of $\tsp$ called many-visits TSP ($\mvTSP$) where a request $r(v)\in \mathbb{Z}_+$ is given for each city $v$ and a single salesperson needs to visit each city $r(v)$ times and return back to his starting point. A combination of $\mtsp$ and $\mvTSP$ called many-visits multiple TSP $(\MVmTSP)$ was studied by Bérczi, Mnich, and Vincze \cite{many-visit} where the authors give approximation algorithms for various variants of $\MVmTSP$.

    In this work, we show a simple linear programming (LP) based reduction that converts a $\mtsp$ LP-based algorithm to a LP-based algorithm for  $\MVmTSP$ with the same approximation factor. We apply this reduction to improve or match the current best approximation factors of several variants of the $\MVmTSP$. Our reduction shows that the addition of visit requests $r(v)$ to $\mtsp$ does \emph{not} make the problem harder to approximate even when $r(v)$ is exponential in number of vertices. 
    To apply our reduction, we either use existing LP-based algorithms for $\mtsp$ variants or show that several existing combinatorial algorithms for $\mtsp$ variants can be interpreted as LP-based algorithms. This allows us to apply our reduction to these combinatorial algorithms as well achieving the improved guarantees. 
\end{abstract}

\newpage
\section{Introduction}
The traveling salesperson problem ($\tsp$) is a fundamental problem in combinatorial optimization. Given a complete graph on $n$ vertices and non-negative edge costs that satisfy the triangle inequality, the goal is to find a Hamiltonian cycle of minimum cost that visits all vertices. $\tsp$ and its variants have been at the forefront of development of algorithms, in theory as well as practice. From an approximation algorithmic perspective, Christofides \cite{Christofides} and Serdyukov \cite{serdyukov}, gave a $3/2$-approximation algorithm for $\mathrm{TSP}$ which was recently improved to roughly $3/2 - 10^{-36}$ by Karlin, Klein, and Oveis-Gharan \cite{KarlinKG21}.


In this work, we aim to consider the multiple visit versions of $\tsp$ as well as many of its variants. In the multiple visit version of $\tsp$, which we call $\mvTSP$, we are given a requirement $r(v)\in \mathbb{Z}_+$ for each vertex $v\in V$ and the goal is to find a closed walk that visits each vertex exactly $r(v)$ times. Simply, introducing $r(v)$ \emph{copies} of each vertex $v\in V$ and solving the $\tsp$ instance in the corresponding semi-metric, it is easy to see that any $\rho$-approximation for $\tsp$ gives a $\rho$-approximation algorithm for $\mvTSP$. Unfortunately, this reduction is not polynomial time since the input size is logarithmic in $\max_{v}r(v)$ while the algorithm takes time polynomial in $\max_{v} r(v)$. This raises an important question:

 Is there a polynomial-time reduction that implies that a $\rho$-approximation for $\tsp$ gives a $\rho$-approximation for $\mvTSP$?

We ask the same question for variants of the $\tsp$ problem, in particular, for the variants inspired by the classical vehicle routing problem. An extension of TSP is multiple TSP, which we call $\mathrm{mTSP}$, where there is a specified number of salespersons $k$ and the goal is to find $k$ minimum cost cycles such that all vertices are visited by one salesperson. There are several variations depending on whether the salesperson start at a fixed set of depot vertices $D$ and  whether or not all salesperson need to be used. We refer the reader to a survey by Bektas \cite{Bektas} detailing different variants, applications, and several algorithms for $\mathrm{mTSP}$. The multi-visit version of  $\mathrm{mTSP}$ that we call $\MVmTSP$ is again defined similarly: we are given a graph, edge costs, a visit function $r:V\rightarrow \mathbb{Z}_+$, an integer $k$, and possibly a set of $k$ depots and the goal is to find $k$ minimum cost closed walks so that each vertex is visited $r(v)$ times. There are several variants depending on whether there are depots, if all salesperson need be used, and if the demands for each vertex can be satisfied by multiple salesperson. Approximation algorithms for many of these variants were studied recently~\cite{many-visit}. We give a detailed description of the different variants in Section \ref{Section:description}.


\subsection{Our Results and Contributions}

Our main result is to show that there is a polynomial time algorithm that implies any $\rho$-approximation algorithm for $\tsp$ that is \emph{linear programming based} implies a $\rho$-approximation algorithm for $\mvTSP$. By LP-based algorithms, we imply any algorithm that ensures that the objective value of the output solution is at most $\rho$ times the objective of the classical Held-Karp LP relaxation for the $\tsp$. 

\begin{restatable*}{theorem}{THMmv}
\label{theorem:mvMtsp}

    If there is a $\rho$-approximation algorithm for $\tsp$ where the $\rho$ guarantee is towards the optimum solution of LP \eqref{LP:tsp} then there is a $\rho$-approximation algorithm for the $\mvTSP$ problem.
\end{restatable*}

We also show that the above reduction also holds for various $\MVmTSP$ as well. This allows to either obtain improved approximation algorithms or match the best known approximation for many of these variants. We outline the improved approximation in Table \ref{table1}. For many variants of $\mtsp$, previously only combinatorial algorithms were known. We first interpret these combinatorial algorithms as LP-based algorithms by showing that they also give a bound on the integrality gap of the standard Held-Karp style relaxations for these variants. For $\tsp$, this part is analogous to showing the classical Christofides' $\frac32$-approximation algorithm also bounds the integrality gap of the Held-Karp relaxation to within the bound of $\frac32$ as was done by Wolsey \cite{Wolsey} and Shmoys and Williamson \cite{shmoys-willaimson-HK}.


\begin{table}\label{table1}

\begin{tabular}{|c|c|c|c|c|}
    \hline
    Depot Restriction  &  SV  & MV Problem Name&  MV   & MV \\ [0.5ex] 
    &   (previous work) & &   (previous work)  & (this work) \\ [0.5ex] 
    \hline\hline
    $\leq$  & $3/2 + \varepsilon$ \cite{VarDepot-3/2}   & $\MVmTSP_0$ & 2 \cite{many-visit} & $\mathbf{2}$  \\
    $=$ & $2$ \cite{many-visit} & $\MVmTSP_+$ & 3 \cite{many-visit} & $\mathbf{2}$  \\
    One Depot, $k$ salesperson &$3/2$ \cite{Frieze}  & $\SDMVTSP_+$ & 3 \cite{many-visit} & $\mathbf{3/2}$ \\
        \hline
\end{tabular}
\caption{Variants where all $k$ depot/salesperson must be used are marked with $=$ in the first column and variants where some tours may be empty are marked with $\leq$. MV stands for multi-visit and SV stands for single-visit. The numbers in the table are all the best approximation factors. }
\end{table}
There is another variant that includes an additional constraint which requires each vertex to be visited by exactly one salesperson. This means that all vertices are visited $r(v)$ times and all tours are vertex disjoint. We cannot apply our technique to the vertex disjoint tours variants, however we are able to get an approximation for the single depot multi-visit multiple TSP problem using ideas from \cite{many-visit}.
\begin{restatable*}{theorem}{depotDisjoint}
        There is a polynomial time algorithm for the single depot multi-visit $\mtsp$ ($\SDMVTSP$) problem with vertex disjoint tours with an approximation factor of $\frac{7}{2}$.
\end{restatable*}
Next we list a result for one non-depot $\mtsp$ variant where we are required to use all $k$ salespersons which we call the unrestricted $\mtsp_+$ problem. Here we only get an improved result for a single visit variant and are not able to apply the reduction to its multi-visit variant. The previous best approximation factor was $4$ which was shown by Bérczi et al \cite{many-visit}.  
\begin{restatable*}{theorem}{Umtsp}
    There is a polynomial time algorithm for the unrestricted $\mtsp_+$ problem with an approximation factor of $2$. 
\end{restatable*}

\subsubsection{Overview of Technique and Results}\label{Section:Overview}

We now give a overview of the main technique using the example of $\mvTSP$. The same techniques apply to rest of the variants with some additional modifications. We assume that there is a $\rho$-approximation for $\tsp$ which is LP-based. As previously mentioned, the running time of the algorithm for $\mvTSP$ needs to be polynomial in $\max_{v \in V}\log r(v)$ and $n$. A simple exponential time approximation algorithm for $\mvTSP$ is to make $r(v)$ copies of each vertex $v$ and apply a $\mathrm{TSP}$ $\rho$-approximation algorithm to this graph.  
On the other hand, if $\max_{v \in V}r(v)$ was polynomial in $n$ then we would get a $\rho$-approximation polynomial time algorithm. Our main technique is to use a LP relaxation of $\mvTSP$ to fix certain edges in our solution (without taking a loss in the objective) and construct a new instance where the visit requirement of each vertex is polynomial. We then apply the simple reduction to $\mathrm{TSP}$ we described above. We note that our reduction relies on the connection between the LP relaxations of $\mvTSP$ and $\mathrm{TSP}$: the LP relaxations only differ in that $\mathrm{TSP}$ requires every vertex has degree $2$ while $\mvTSP$ has degree $2r(v)$. As a result our reduction is limited in that we cannot use any algorithm for $\mathrm{TSP}$ but only an algorithm that has a guarantee towards the LP relaxation of $\mathrm{TSP}$. 

For many variants that we consider, a challenge then arises. Several of the existing algorithms give approximation guarantees as compared to the integral solution and do not compare the algorithm's solution to the cost of the linear programming relaxation. For these problems, we first formulate a Held-Karp style LP relaxation and either show that an existing algorithm has an approximation guarantee relative to the LP value or give a new algorithm which has a guarantee towards the LP value. For this we use characterizations of matroid intersection polytope which we apply to constrained spanning trees and related problems in Section~\ref{sec:tree}.

To illustrate our technique in more detail, let us return to the $\mvTSP$ problem. All missing proofs for this result will appear in Appendix \ref{Section:TSP-proofs} and also appear more generally in our general framework in Section \ref{Section:framework}.
We use the following standard Held-Karp LP relaxation for $\tsp$ which we call LP \eqref{LP:tsp},
\begin{align}
\text{minimize   } &\sum_{e \in E}c_ex_e \label{LP:tsp} \\
\text{s.t.   } & x(\delta(v)) = 2 &\forall v \in V \nonumber \\
 & x(\delta(S)) \geq 2 &\forall S \subseteq V \nonumber \\
 & 0 \leq x_e \leq 1 & \forall e \in E. \nonumber
\end{align}

The following linear program is a relaxation for $\mvTSP$ that generalizes the Held-Karp LP,
\begin{align}
\text{minimize   } &\sum_{e \in E}c_ex_e \label{LP:tsp-multi} \\
\text{s.t.   } & x(\delta(v)) = 2r(v) &\forall v \in V \nonumber \\
 &  x(\delta(S))  \geq 2 &\forall S \subset V \nonumber \\
 &  x_e \geq 0 & \forall e  \in  \nonumber E .
\end{align}
We need the following lemma which shows that the simple reduction from $\mvTSP$ to $\tsp$ is polynomial time when the visit requests $r(v)$ are polynomial in $n$. Moreover, the reduction maintains the approximation factor of the LP relaxation based algorithm used for $\tsp$. The reduction basically relies on replacing each vertex with $r(v)$ copies and then applying the LP-based algorithm. 
\begin{lemma} \label{lemma:smallR-tsp}
    If there is a $\rho$-approximation algorithm for $\tsp$ where the $\rho$ guarantee is towards the optimum solution of LP \eqref{LP:tsp} then there exists an algorithm that given a solution $y$ to LP \eqref{LP:tsp-multi} outputs a closed walk $T: E \to \mathbb{Z}$ satisfying
    $\sum_{e \in C}T(e)c_e \leq \rho \sum_{e \in E}c_ey_e$ with a run-time polynomial in $\max_{v \in V}r(v)$ and $n$. 
\end{lemma}

Now we show how to use Algorithm in Lemma~\ref{lemma:smallR-tsp} for a general instance where $r$ is not polynomially bounded. This algorithm (Algorithm~\eqref{Algorithm:MVTSP}) solves LP \eqref{LP:tsp-multi} and fixes edges in the solution that are integrally set and reduce the visit requests accordingly. Finally, the reduced visits are polynomial in size so we can then apply Lemma~\ref{lemma:smallR-tsp}. One has to carefully verify that the \emph{reduced} linear programming solution is a feasible solution to the LP relaxation for the reduced instance which can be done by verifying the constraints carefully. 
\begin{algorithm}[H]
\caption{$\mvTSP$ Reduction Algorithm}
\DontPrintSemicolon
\label{Algorithm:MVTSP}
\KwInput{$G=(V,E), c: V \times V \to \mathbb{R}_{\geq 0},r: V \to \mathbb{Z}$}
\setcounter{AlgoLine}{0}
\KwOutput{an integral solution to LP \eqref{LP:tsp-multi}}

Solve LP \eqref{LP:tsp-multi} to get solution $x^*$ .\;

For all edges $e$ let $\tilde{x}_e = x_e - 2k_e$ such that $k_e = 0$ if $x_e \leq 4$ and otherwise $k_e$ is set so that $2 \leq \tilde{x}_e < 4$ and $k_e \in \mathbb{Z}$. Define a function $\tilde{r}: V \to \mathbb{Z}$ where $\tilde{r}(v) = r(v) - \sum_{e \in \delta(v)}k_e$.\;

Use Lemma \eqref{lemma:smallR-tsp} with solution $\tilde{x}$ on instance $G, \tilde{r}$ to get $T:E \to \mathbb{Z}$.\;

Increase the number of times each edge is used in the previous step by $2k_e$ and return the resulting solution.\;
\end{algorithm}
The next two claims show that Step 3 of the Algorithm \eqref{Algorithm:MVTSP}  is valid.  

\begin{claim} \label{claim:rbound-tsp}
The new visit function $\tilde{r}$ satisfies $1 \leq \tilde{r}(v) \leq 2n$ for all $v \in V$.
\end{claim}
\begin{proof}
    For $v \in V$ we have,
    \begin{align*}
        \tilde{r}(v) &= r(v) - \sum_{e \in \delta(v)}k_e \\
        &= r(v) - \sum_{e \in \delta(v)} \frac{x_e - \tilde{x}_e}{2} = \frac{1}{2}\sum_{e \in \delta(v)}\tilde{x}_e.
    \end{align*}
    If all $e \in \delta(v)$ satisfy $x_e \leq 4$ then $\tilde{r}(v) = r(v) \geq 1$. Otherwise, the lower bound follows since for $x_e > 4$ we have $\tilde{x}_e \geq 2$. The upper bound follows since $\tilde{x}_e \leq 4$ for all $e \in E$.
\end{proof}
\begin{claim} \label{claim:feasible-TSP}
    The solution $\tilde{x}$ is feasible for LP \eqref{LP:tsp-multi} with graph $G$ and $\tilde{r}$.
\end{claim}
\begin{proof}
    We have $\tilde{x} \geq 0$ and $\sum_{e \in \delta(v)}\tilde{x}_e = \sum_{e \in \delta(v)} x_e - 2k_e = 2(r(v) - \sum_{e \in \delta(v)}k_e ) = 2\tilde{r}(v)$. For any set $S \subset V $, if $x_e \leq 4$ for all $e \in \delta(S)$ then $\tilde{x}(\delta(S)) = x(\delta(S)) \geq 2$. Otherwise if there is an edge $e \in \delta(S)$ such that $x_e > 4$ we get $\tilde{x}(\delta(S)) \geq \tilde{x}_e \geq 2$ by definition of $\tilde{x}$. 
\end{proof}
Now we get the following theorem which is our main result.
\THMmv

As a corollary of Theorem \ref{theorem:mvMtsp} we get the following by applying the work of Karlin, Klein, and Oveis-Gharan \cite{KarlinKG22-LP}.
\begin{corollary}
 There is an approximation algorithm for the $\mvTSP$ problem with an approximation factor less than $\frac{3}{2} - 10^{-36}$.
\end{corollary}
 We also apply this reduction to reduce different variants of $\MVmTSP$ to $\mathrm{mTSP}$. One variation of $\mathrm{mTSP}$ is single depot multiple TSP, $\SDTSP$, which was studied by Frieze \cite{Frieze}. For $\SDTSP$ we are given a graph with edge costs, one depot vertex, and an integer $k$ and the goal is to find $k$ cycles that contain the depot so that all non-depot vertices are contained in exactly one cycle. In this paper we also introduce the multi-visit version of this problem called $\SDMVTSP$ which is defined similarly as $\SDTSP$ except each non-depot vertex needs to be visited $r(v)$ times. We also apply our technique to reduce $\SDMVTSP$ to $\SDTSP$. 

\subsection{Related Work}
The variant of $\tsp$ when a visit request $r(v)$ is given for each vertex $v$ is called $\mvTSP$. There is also a variant of $\mvTSP$ called path $\mvTSP$ where instead of a closed walk the goal is to find a walk between given vertices $s \neq t$ so that  vertices $v \notin \{s, t\}$ are visited $r(v)$ times. Bérczi, Mnich, Vincze \cite{BercziMV22} gave a $\frac{3}{2}$  for both path $\mvTSP$ and $\mvTSP$. 

The variant of $\tsp$ where multiple salesperson are used is usually referred to as $\mtsp$. Frieze shows a $\frac{3}{2}$-approximation for a variant where $k$ salesperson are required to start a fixed depot vertex vertex $v_1$. Frieze's algorithm  generalizes the  Christofides-Serdyukov algorithm \cite{Christofides, serdyukov} for $\tsp$. The $\mtsp$ problem is a relaxation of the vehicle routing problem (VRP). In VRP a set of $k$ vehicles need to visit a set of customers with known demands while starting and ending at a fixed depot vertex. Further, the set of vehicles have a vehicle capacity which limits the total demand each vehicle can serve. If the vehicle capacity is sufficiently large so that the vehicles are not restricted by the demands then VRP is equivalent to $\mtsp$. Thus there are several works for VRP that apply ideas from $\tsp$ algorithms such as a paper by Christofides, Mingozzi, Toth \cite{christofides-VRP} where the authors give exact VRP algorithms based on finding min cost trees. 

A different version of $\mtsp$ is when the different salesperson are required to start from different depot vertices. Given a set of $k$ depot vertices the goal is to find at most $k$ minimum cost cycles such that each vertex contains exactly one depot and all vertices are contained in exactly one cycle. Rathinam, Sengupta, and Darbha \cite{rathinam2007resource} showed a $2$-approximation algorithm for this problem which was then improved to $2 - \frac{1}{k}$ by Xu, Xu, and Rodrigues \cite{XuR15-fixed}. Xu and Rodrigues \cite{XuR15-fixed} showed a $\frac{3}{2}$-approximation when the number of depots $k$ is constant and very recently Deppert, Kaul, and Mnich \cite{VarDepot-3/2} showed a $\frac{3}{2}$-approximation for arbitrary $k$.

The $\mtsp$ problem with depots can be generalized further when $m$ depots are available and there are $k$ salesperson satisfying $k \leq m$. Both Malik, Rathinam, and Darbha \cite{malik2007approximation} and Carnes and Shmoys \cite{carnes2011primal} gave $2$-approximations for this problem. Later Xu and Rodrigues \cite{xu2017extension} gave a $(2 - \frac{1}{2k})$-approximation. The algorithm in \cite{XuR15-fixed} can be adapted to this case to get a $\frac{3}{2}$-approximation when $m$ is constant. 

Bérczi, Mnich, and Vincze \cite{many-visit} considered various problems that have both the constraints of $\mtsp$ and $\mvTSP$ which are referred to as $\MVmTSP$.  They consider $8$ different variants of $\MVmTSP$ and show equivalencies among some of the $8$ variants. Additionally, they give constant factor approximations for the different variants using many ideas from previous $\tsp$ algorithms such as tree doubling.

\section{Preliminaries} 
A graph $G = (V, E)$ is defined on vertex set $V$ and edge set $E$ which we will always take to be the complete graph in this paper. For sets $A, B \subseteq V$ we denote by $E(A, B) \subseteq E$ edges that have one endpoint in $A$ and one endpoint in $B$. We use $E(A)$ as a shorthand for $E(A, A)$ and $\delta(A) = E(A, V - A)$ meaning $\delta(A)$ is the set of edges with exactly one endpoint in $A$. For a single vertex $v$ we write $\delta(v)$ instead of $\delta(\{v\})$ to denote its set of neighbors. The degree of a vertex $v$ is denoted by $d(v)$ which is the number of edges incident to that vertex meaning $d(v) = |\delta(v)|$. We note that any loops on a vertex contribute $2$ to the degree. Additionally for a set of edges $T \subseteq E$, we use $d_T(v)$ to denote the number of edges in $T$ that contain $v$. Throughout the paper we use LPs whose variables correspond to edges of the graph and for LP variable $x \in \mathbb{R}^{|E|}$ we use $x(T) = \sum_{e \in T}x_e$ for all $T \subseteq E$.

We also use the notion of a matroid in this paper. A matroid $\mathcal{M}$ is defined by a ground set $E$ and a collection of independent sets $\mathcal{I} \subseteq 2^{E}$ satisfying three properties 
\begin{enumerate}
\item $\emptyset \in \mathcal{I}$ 
\item If $A \in \mathcal{I}$ then $B \in \mathcal{I}$ for all $B \subseteq A$ 
\item If $A, B \in \mathcal{I}$ with $|A| < |B|$ then there exists $x \in B - A$ so that $A \cup \{x\} \in \mathcal{I}$.
\end{enumerate}
 An independent set of maximum cardinality is called a base. We use two specific matroids in this paper: partition matroids and graphic matroids. A partition matroid  is defined by a partition of the ground set $E = P_1 \dot\cup \dots \dot\cup P_k$ each with a capacity $c_i \leq |P_i|$ and a set $S \in \mathcal{I}$ if $|S \cap P_i| \leq c_i$ for all $i = 1,\ldots,k$. A graphic matroid is defined on a graph $G$ with the set of edges as the ground set and a set $T \subseteq E$ is independent if the set of edges $T$ is acyclic in $G$. All matroids $\mathcal{M}$ have a rank function $r: 2^E \to \mathbb{Z}$ which is defined as  $r(S) = \max_{A \subseteq S} \{ |A| | A \in \mathcal{I} \} $. It is well known that the convex hull of indicator vectors of independent sets in a matroid is described by $\{ \mathbf{x} \geq 0 \in R^{|E|} | x(S) \leq r(S) \forall S \subseteq E \}$ and for matroids $\mathcal{M}_1 = (E, \mathcal{I}_1), \mathcal{M}_2 = (E, \mathcal{I}_2)$ with rank functions $r_1, r_2$ the convex hull of the indicator vectors of $\mathcal{I}_1 \cap \mathcal{I}_2$ is given by $\{ \mathbf{x} \geq 0 \in R^{|E|} | x(S) \leq \min(r_1(S),r_2(S))  \forall S \subseteq E \}$. Moreover, both the matroid and matroid intersection polytopes are TDI (totally dual integral). We refer the reader to \cite{schrijver-book} for more details on matroids and matroid polytopes.

\section{Problem Description} \label{Section:description}
In this section we formally describe each of the problems and the requirements of feasibility. We use the same names and notations for the problem names as \cite{many-visit}. Throughout the paper we will use $n$ to denote the number of vertices in the graph that is the input to each problem and $c: E \times E \to \mathbb{R}_{\geq 0}$ to denote the cost function. The cost function $c$ is a semi-metric meaning it is symmetric and satisfies triangle inequality, but does not satisfy $c_{vv} = 0$. In particular, we have the following 
\begin{enumerate}
    \item Symmetry: $c_{uv} = c_{vu}$ for all $u, v \in V$
    \item Triangle Inequality: $c_{uv} \leq c_{ux} + c_{xv}$ for all $u, v, x \in V$. 
\end{enumerate}
We observe that the triangle inequality implies for all $u,v \in V$ $c_{vv} \leq 2c_{uv}$. This means our algorithms are allowed to use loops to satisfy the visit requirement and will pay for those loops. We recall that a loop is counted twice in the degree of a vertex so taking a loop on a vertex counts as one visit to that vertex. In many single visit variants such as $\mathrm{TSP}$ using loops violates feasibility so for most single visit variants we may assume $c_{vv} = 0$ and $c$ is a metric. There are a few single visit variants that use loops and this will be specified below.

First we describe the simpler variants. 
\begin{enumerate}
    \item We call the standard traveling salesperson problem as $\tsp$. For $\tsp$, we are given a complete graph $G$ on a vertex set $V$ and the goal is to find a Hamiltonian cycle of minimum cost. 
     \item We call the multi-visit TSP problem as $\mvTSP$. In  $\mvTSP$, we are given a complete graph $G$ on a vertex set $V$ and a visit function $r: V \to \mathbb{Z}$ satisfying $r(v) \geq 1$ for all $v \in V$. The goal is to find a minimum cost closed walk such all vertices are visited exactly $r(v)$ times.
    \item We call the multiple TSP problem as $\mtsp$. There are $4$ variants of $\mtsp$ depending on $2$ parameters.

    \begin{enumerate}[label*=\arabic*.]
    \item The first is whether or not all salespersons are used. In $\mathrm{mTSP}_+$, we are given a complete graph $G$ on a vertex set $V$ and a number $k \geq 1$ and the goal is to find exactly $k$ cycles such that every vertex is contained in exactly one cycle. In $\mathrm{mTSP_0}$ we have the same inputs and the goal is to find at most $k$ cycles so that every vertex is contained in exactly one cycle.
    \item The second is whether or not we have depot vertices. If there are depot vertices then any pair of depots cannot be in the same cycle. If there are no depots then the problem is called \emph{unrestricted}.
    \end{enumerate}
    From these two parameters we get the following $4$ problems: unrestricted $\mathrm{mTSP}_+$, unrestricted $\mathrm{mTSP}_0$, $\mathrm{mTSP}_+$, and $\mathrm{mTSP}_0$. Both $\mathrm{mTSP}_+$, and $\mathrm{mTSP}_0$ are allowed to use loops.
    
    \end{enumerate}
In this paper we mainly study variants that are a mix of the multiple TSP and multi-visit TSP problems which we call  $\MVmTSP$ For $\MVmTSP$, there are $8$ variants of problems that arise from three parameters. The two parameters from $\mtsp$ carry over which are whether or not all salesperson are used and whether or not there are depot vertices. In addition to these two parameters, there is also a parameter that imposes a restriction on whether the visit requirements for a vertex need to satisfied by one salesperson or if different salesperson in total can satisfy the visit requirements. Variants where the different tours are required to be vertex disjoint are called \emph{vertex disjoint} tours and variants where different tours are allowed to intersect are called \emph{arbitrary}. Using these $3$ parameters we get $8$ problems some of which include unrestricted $\mathrm{mTSP}_+$ with arbitrary tours, $\mathrm{mTSP}_+$ with vertex disjoint tours, and $\mathrm{mTSP}_0$ with vertex disjoint tours.

We note that some of the variants among all $8$ possibilities are equivalent, some can be reduced to another in only one direction, and others cannot be reduced to each other in either direction. We refer the reader to \cite{many-visit} which explains the connections between the problems in detail and gives examples. We now describe the problems that we consider in this paper.
\begin{enumerate}
    \item In $\mathrm{MV\mbox{-}mTSP_+}$ with arbitrary tours, we are given a complete graph $G$ on a vertex set $V$ and a subset of depots $D \subseteq V$ with $|D| = k$. The goal is to find exactly $k$ closed walks such that all non-depot vertices $v$ are visited $r(v)$ times and each closed walk contains exactly one depot. 
    \item In unrestricted $\mtsp_+$ with arbitrary tours, we are given a complete graph $G$ on a vertex set $V$ and an integer $k$. The goal is to find $k$ cycles so that every vertex is contained in one cycle. We note that here a single loop on a vertex is a valid cycle.
\end{enumerate}
Next we describe variants where there is one depot, but we have $1 \leq k \leq n-1$ salesperson available. 
\begin{enumerate}
    \item In $\SDTSP_+$, we are given a complete graph and an integer $1 \leq k \leq n-1$. The goal is to find exactly $k$ cycles containing at least $3$ vertices so that all cycle contain the depot vertex $v_1$ and every other vertex $v \neq v_1$ is contained in exactly one cycle. We note that if cycles with two vertices were allowed, meaning $v_1, v, v_1$ is valid, then Frieze \cite{Frieze} shows a reduction to the problem where each cycle has at least $3$ vertices. 
    \item In $\SDMVTSP_+$ with arbitrary tours, we are given a complete graph, an integer $1 \leq m \leq n-1$, and a visit function $r:V - v_1 \to \mathbb{Z}$. The goal is to find exactly $k$ closed walks starting at the depot vertex $v_1$ so that all non-depot vertices $v$ are visited $r(v)$ times.
    \item In $\SDMVTSP_+$ with vertex disjoint tours, we are given a complete graph, an integer $1 \leq m \leq n-1$, and a visit function $r:V - v_1 \to \mathbb{Z}$. is to find exactly $k$ closed walks starting at the depot vertex so that all non-depot vertices $v$ are visited $r(v)$ times and any two closed walks only intersect at the depot vertex. 
\end{enumerate}

\section{General Framework} \label{Section:framework}
In this section we outline a general framework that we will apply to the various variants of multi-visit TSP problems. The following LP will be our template for the LP relaxation of the single visit variants of the problems we will consider. We note that only $\mvTSP$ and $\MVmTSP_+$ with arbitrary tours are the only problems that fall \emph{exactly} into the framework. Other problems fall very closely in the framework and for those problems we follow the template given in this section and use parts of the proofs given here.  For $D \subseteq V$ we write the following LP
\begin{align}
\text{minimize   } &\sum_{e \in E}c_ex_e \label{LP:general-single}\\
\text{s.t.   } & x(\delta(v)) = 2 &\forall v \in V \nonumber \\
 & x(\delta(S \cup D)) \geq 2 &\forall S \subseteq V - D \nonumber \\
 & x(E(D, D)) = 0 & \nonumber \\
 & 0 \leq x_e \leq 2 & \forall e \in E. \nonumber
\end{align}
We now show its corresponding multi-visit variant for a visit function $r: V - D \to \mathbb{Z}$
\begin{align}
\text{minimize   } &\sum_{e \in E}c_ex_e \label{LP:general-multi}\\
\text{s.t.   } & x(\delta(v)) = 2r(v) &\forall v \in V - D \nonumber \\
& x(\delta(v)) = 2 &\forall v \in D \nonumber \\
 & x(\delta(S \cup D)) \geq 2 &\forall S \subset V - D \nonumber \\
 & x(E(D, D)) = 0 & \nonumber \\
 & x_e \geq 0 & \forall e  \in E. \nonumber
\end{align}
We note that we can set $D = \emptyset$ in the above linear programs which gives us the result for $\mvTSP$ shown in Section \ref{Section:Overview}.

The following lemma allows us to relate the LP relaxation of the multi-visit problem to its corresponding single visit variant. 
\begin{lemma} \label{lemma:solution-construction}
    Let $G$ be a complete graph with edge set $E$, $c$ be a cost function on its edges, and $r: V - D \to \mathbb{Z}$ be a integer valued function on the vertex set. Let $G^r = (V^r, E^r)$ be a complete graph on vertex set $V^r$ that has $r(v)$ copies of each vertex $v \in V - D$ and one copy of each vertex $v \in D$. We extend the cost $c$ function that assigns cost $c_e$ for $e = \{u, v\} \in E$ to all edges $\{u_i, v_j\} \in E^r$. Given a solution $x$ to LP \eqref{LP:general-multi} we can construct a feasible solution $z$ to LP \eqref{LP:general-single} on graph $G^r$ satisfying $c^Tx = c^Tz$.
\end{lemma}
\begin{proof}
    Let $e' = \{u_i, v_j\} \in E^r$ and $e = \{u, v\}$ are the original copies of $u_i, v_j$ in $V$, then we set $z_{e} = \frac{x_e}{r(u)r(v)}$ and if either $u, v$ are in $D$ we think of $r(u) = 1$ or $r(v) = 1$ while defining $z$. For the last property we have, $\sum_{e \in E^r}c_e z_e = \sum_{e =\{u, v\} \in E} c_e r(u)r(v) z_e = \sum_{e \in E} c_e x_e $. We observe that the solution $x$ is constructed by distributing the degree of each vertex $v \in V - D$ uniformly to all $r(v)$ copies so we get that $x(\delta(v_i)) = \frac{z(\delta(v))}{r(v)} = 2$ and for $v \in D$ $x(\delta(v)) = z(\delta(v)) = 2$.  We get that for any $e \in E(D, D)$ we have $z_e = x_e = 0$ so $z(E(D, D)) = 0$. Finally $0 \leq z \leq 2$ follows since $z_e \leq 2$ if and only if $x_e \leq 2r(u)r(v)$ which follows since $x_e \leq 2\min(r(u), r(v))$.

    We show that $x$ satisfies $x(\delta(D \cup T)) \geq 2$ for all $T \subset V^r - D$. Let $k$ be the number of vertices $v \in V$ such that there exist copies of $v_i, v_j$ such that $T$ contains exactly one of $v_i, v_j$. We show that $x(\delta(D \cup T)) \geq 2$  exists by induction on $k$. If $k = 0$, then we take $S \subset V - D$ to be the set acquired by taking the original copy $v$ of each vertex $v_i \in T$ and we get that $z(\delta(T \cup D)) = x(\delta(S \cup D)) \geq 2$. If $k > 0$ there exits a vertex $v \in V - D$ such that $T$ does not contain all $r(v)$ copies of $v$. Let $B = V^r - (T \cup D)$, $T(v) = T \cap \{v_1, \ldots, v_{r(v)} \}$, $B(v) = \{v_1, \ldots, v_{r(v)} \} - T(v)$. For any $v_i \in T(v)$ let $Z_1 = z(E(v_i, D \cup T - T(v)), Z_2 = z(E(v_i, B))$ and we note that the values of $Z_1, Z_2$ do not change depending on the choice of $v_i$ since all copies of vertex $v$ have the same set of neighbors and $z$ values assigned to their edges. Thus we have that $z(E(T(v), D \cup T - T(v))) =  |T(v)|Z_1 $ and $z(E(T(v), B)) = |T(v)|Z_2$. In the case that $Z_2 \geq Z_1$, we observe that $z(\delta(D \cup T - T(v)) = z(\delta(D \cup T)) + z(E(T(v), D \cup T - T(v))) - z(E(T(v), B)) =z(\delta(D \cup T - T(v))) + |T(v)|(Z_1- Z_2) \leq  z(\delta(D \cup T - T(v)))$. Thus we get $z(\delta(D \cup T - T(v)) \geq 2$  by applying induction to $T - T(v)$ implying $z(\delta(D \cup T)) \geq 2$. Similarly if $Z_1 \geq Z_2$, we get that $2 \leq z(\delta(D \cup T \cup B(v)) = z(\delta(D \cup T)) + z(E(B(v), T \cup D)) - z(E(B(v), B - B(v))) = z(\delta(D \cup T)) + |B(v)|(Z_2 - Z_1) \leq  z(\delta(D \cup T \cup B(v)))$ and first inequality follows by applying induction to $T \cup B(v)$.
\end{proof}
Using these LP relaxations and Lemma \ref{lemma:solution-construction} , we can use apply an approximation algorithm for the single visit case to the multi-visit case which will have a run-time depending on $n$ and the visit function $r$. 

\begin{lemma} \label{lemma:general-smallR}
    Let $\mathcal{A}$ be an approximation algorithm that takes a solution $x$ to LP \eqref{LP:general-single} and outputs an integral solution $C \subseteq E$ to LP \eqref{LP:general-single} satisfying $\sum_{e \in C}c_e \leq \rho \sum_{e \in E}c_e x_e$ with a run-time polynomial in $n$. There exists an algorithm that given a solution $y$ to LP \eqref{LP:general-multi} outputs an integral solution $T: E \to \mathbb{Z}$ satisfying $\sum_{e \in E}c_eT(e) \leq \rho \sum_{e \in E}c_ey_e$ with a run-time polynomial in $\max_{v \in V}r(v)$ and $n$. 
\end{lemma}
\begin{proof}
    We will convert the solution $y$ to a solution $x$ to LP \eqref{LP:tsp} on the graph $G^r$ by applying Lemma \ref{lemma:solution-construction} to $y$. Thus we can apply Algorithm $\mathcal{A}$ on $y$ as a solution to $G^r$ which we convert to a solution on $G$ by replacing every edge in the solution with its original corresponding edge in $G$. Let $T: E \to \mathbb{Z}$ be the solution we get. We now show $T$ is feasible for LP \eqref{LP:general-multi}. First for every $v \in V$ we get, $\sum_{e \in \delta(v)}T(e) = 2r(v)$ since the solution in $G^r$ contained $r(v)$ copies of $v$ each with degree $2$. Moreover, for all $S \subset V - D$ we have $\sum_{e \in \delta(S \cup D)}T(e) \geq 2$. For the approximation factor we have, $\sum_{e \in E}T(e)c_e \leq \rho \sum_{e \in E}c_e x_e = \rho \sum_{e \in E}c_e y_e$ where the first equality follows since we $\mathcal{A}$ is a $\rho$-approximation and the second equality follows by Lemma \ref{lemma:solution-construction}. Finally, the run-time follows since $G^r$ has at most $n \max_{v \in V}r(v)$ vertices and $\mathcal{A}$ is a polynomial time algorithm. 
\end{proof}

Using this lemma we get the following polynomial time algorithm. 
\begin{algorithm}[H]
\caption{General Reduction Algorithm}
\DontPrintSemicolon
\label{Algorithm:general}
\KwInput{$G=(V,E), D \subseteq V ,c: V \times V \to \mathbb{R}_{\geq 0},r: V - D \to \mathbb{Z}$}
\setcounter{AlgoLine}{0}
\KwOutput{an integral solution to LP \eqref{LP:general-multi}}

Solve LP \eqref{LP:general-multi} to get solution $x^*$ .\;

For all edges $e$ let $\tilde{x}_e = x_e - 2k_e$ such that $k_e = 0$ if $x_e \leq 4$ and otherwise $k_e$ is set so that $2 \leq \tilde{x}_e < 4$ and $k_e \in \mathbb{Z}$. Define a function $\tilde{r}: V \to \mathbb{Z}$ where $\tilde{r}(v) = r(v) - \sum_{e \in \delta(v)}k_e$.\;

Use Lemma \ref{lemma:general-smallR} with solution $\tilde{x}$ on instance $G, \tilde{r}$ to get $T:E \to \mathbb{Z}$.\;

Increase the number of times each edge is used in the previous step by $2k_e$ and return the resulting solution.\;
\end{algorithm}

The next two claims show that Step 3 of the algorithm is valid.

\begin{claim} \label{claim:rbound}
The new visit function $\tilde{r}$ satisfies $1 \leq \tilde{r}(v) \leq 2n$ for all $v \in V$.
\end{claim}
\begin{proof}
    For $v \in V$ we have,
    \begin{align*}
        \tilde{r}(v) &= r(v) - \sum_{e \in \delta(v)}k_e \\
        &= r(v) - \sum_{e \in \delta(v)} \frac{x_e - \tilde{x}_e}{2} = \frac{1}{2}\sum_{e \in \delta(v)}\tilde{x}_e.
    \end{align*}
    If all $e \in \delta(v)$ satisfy $x_e \leq 4$ then $\tilde{r}(v) = r(v) \geq 1$. Otherwise the lower bound follows since for $x_e > 4$ we have $\tilde{x}_e \geq 2$. The upper bound follows since $\tilde{x}_e \leq 4$ for all $e \in E$.
\end{proof}
\begin{claim} \label{claim:general-feasible}
    The solution $\tilde{x}$ is feasible for LP \eqref{LP:general-multi} with graph $G$ and $\tilde{r}$.
\end{claim}
\begin{proof}
    We have $\tilde{x} \geq 0$ and $\sum_{e \in \delta(v)}\tilde{x}_e = \sum_{e \in \delta(v)} x_e - 2k_e = 2(r(v) - \sum_{e \in \delta(v)}k_e ) = 2\tilde{r}(v)$. For any set $S \subset V - D$, if $x_e \leq 4$ for all $e \in \delta(S \cup D)$ then $\tilde{x}(\delta(S \cup D)) = x(\delta(S \cup D)) \geq 2$. Otherwise if there is an edge $e \in \delta(S \cup D)$ such that $x_e > 4$ we get $\tilde{x}(\delta(S \cup D)) \geq \tilde{x}_e \geq 2$ by definition of $\tilde{x}$. Finally we have $\tilde{x}_e = x_e = 0$ for all $e \in E(D, D)$ implying $\tilde{x}(E(D, D)) = 0$. Moreover, $T$ contains no edges between vertices in $D$ since $k_e = 0$ for all $e \in E(D, D)$ and $T'(e) = 0$ for all $e \in E(D, D)$.
\end{proof}

\begin{theorem} \label{theorem:general}
    Let $x^*$ be the optimal LP solution to LP \eqref{LP:general-multi} and $\rho$ be the approximation factor of algorithm $\mathcal{A}$ whose guarantee is relative the value of LP \eqref{LP:general-single}. Then Algorithm \eqref{Algorithm:general} outputs a feasible integral solution to LP \eqref{LP:general-multi}, $T: E \to \mathbb{Z}$ satisfying $\sum_{e \in E}T(e)c_e \leq \rho c^Tx^*$ and runs in time polynomial in $n$.
\end{theorem}
    \begin{proof}
        Let $T'$ be the solution we get from the third step before we increase each edge by $2k_e$. By Claim \ref{claim:rbound} and Lemma \ref{lemma:general-smallR} we have that $\tilde{x}$ satisfies 
        $ \sum_{e \in E} T'(e)c_e \leq  \rho \sum_{e \in E}c_e \tilde{x}_e $. 
        Let $S = \{e \in E | k_e > 0  \}$ be the set of edges then we have that $\sum_{e \in E}T(e)c_e = \sum_{e \in E}T'(e)c_e  + 2\sum_{e \in S} k_e c_e \leq \rho \sum_{e \in E}c_e \tilde{x}_e  + 2\sum_{e \in S} k_e c_e \leq \rho c^Tx^*$. For the run-time, we get have that Step $3$ runs in time polynomial in $n$ by  Claim \ref{claim:rbound} and Lemma \ref{lemma:general-smallR} and the rest of the steps are clearly polynomial time in $n$. Finally, $T$ is a feasible integral solution to LP \eqref{LP:general-multi} since $T'$ satisfies the cut constraints which implies $T$ also satisfies the cut constraints since $T(e) \geq T'(e)$ for all $e$ and by definition $T$ satisfies degree constraints. 
    \end{proof}

\section{Tree Characterizations}\label{sec:tree}
Many algorithms for $\tsp$ variants work with a tree with some specified structure and a LP characterization of the tree is needed to bound the approximation guarantee. In this section we characterize the polytopes of three different trees. We first characterize connected graphs that have fixed degree $2k$ on vertex $v_1 \in V$.
We define $\kappa(S)$ as the number of components in the graph $(V, S)$ for all $S \subseteq E$,
    \begin{align} 
        \text{minimize   } &\sum_{e \in E}c_ez_e \label{LP:tree} \\
        \text{s.t.   } & z(S) \geq \kappa(\overline{S}) -1   & \forall S \subseteq E \nonumber \\
        & z(\delta(v_1)) = 2k & \nonumber \\
         & z_e \geq  0 & \forall e \in E \nonumber.
    \end{align}
    \begin{claim} \label{claim:LP:tree}
        Let $T^*_{2k}$ be a minimum cost spanning tree among all spanning trees in $G$ that have degree $2k$ on vertex $v_1$. Then we have that the indicator vector $T^*_{2k}$ is an optimal solution to LP \eqref{LP:tree}.
    \end{claim}
    \begin{proof}
    
    First we define matroid $\mathcal{M}_1$ as the dual of the graphic matroid on graph $G$. Next we define $\mathcal{M}_2$ as a partition matroid with parts $P_1,P_2, \ldots P_{|E|}$ where $P_1$ contains edges incident to $v_1$ and has capacity $d(v_1) - 2k$, and the remaining edges go in a unique $P_j$ with capacity $1$. Thus common independent sets of $\mathcal{M}_1, \mathcal{M}_2$ are sets $R \subseteq E$ such that $R$ has at most $d(v_1) - 2k$ edges incident to $v_1$ and $E - R$ contains a spanning tree. Moreover, an optimal solution must be also satisfy that $E - R$ is a spanning tree. By the matroid intersection theorem, the polytope given by constraints is $\{ x \geq 0 | x(S) \leq r_{\mathcal{M}_i}(S), \forall i \in [2] \text{ and } S \subseteq E \}$ totally-dual integral and therefore integral. Turning an inequality to equality in a TDI system maintains the TDI property so we can restrict our polytope to common independent sets of $\mathcal{M}_1, \mathcal{M}_2$ that have degree exactly $d(v_1) -2k$ on $v_1$. Then we observe that $E - T^*_{2k}$ is an optimal solution to $\max_{I \in \mathcal{I}_1 \cap \mathcal{I}_2}c(I)$ and is also an optimal solution to the following LP
    \begin{align} 
        \text{maximize } &\sum_{e \in E}c_ex_e \label{LP:matroid} \\
        \text{s.t.   } & x(\delta(v_1)) = d(v_1) - 2k   & \nonumber \\
        & x(S) \leq r_{\mathcal{M}_1}(S) & \forall S \subset E \nonumber \\
        & x_e \leq 1 & \forall e \in E \nonumber.
    \end{align}
Here we have that $r_{\mathcal{M}_1}(S) = |S| - \kappa(\overline{S}) + 1$  since $\mathcal{M}_1$ is the dual matroid to the graphic matroid. If we negate the objective of LP \eqref{LP:matroid}, change it to a minimization problem, add $\sum_{e \in e}c_e $ to the objective, and make the variable change $z_e = 1 - x_e$ we get LP \eqref{LP:tree}. This is true since the following hold for all $S \subseteq E$
\begin{enumerate}
    \item $\sum_{e \in E}c_e - \sum_{e \in E}c_ex_e = \sum_{e \in e}c_ez_e$
    \item $x(\delta(v_1)) = d(v_1) - 2k \iff 2k = d(v_1) - x(\delta(v_1)) \iff 2k = z(\delta(v_1))$ 
    \item $x(S) \leq |S| - \kappa(E - S) + 1 \iff \kappa(\overline{S}) - 1 \leq |S| - x(S) \iff \kappa(\overline{S}) - 1 \leq z(S)$.
\end{enumerate}

If $x^*$ is an optimal solution to LP \eqref{LP:matroid} then $z^* = \mathbf{1} - x^*$ is an optimal solution to LP \eqref{LP:tree}, so $T_{2k}^*$ is an optimal solution to LP \eqref{LP:tree} since $E - T_{2k}^*$ is an optimal solution to LP \eqref{LP:matroid}.
 \end{proof}
Similar to the previous LP, we now give the LP formulation of spanning trees who have fixed degree $2k \leq n-1$ on a fixed vertex $v_1$. We use the following LP where $n$ is the number of vertices in the graph 
    \begin{align} 
        \text{minimize   } &\sum_{e \in E}c_ez_e \label{LP:tree-unrestricted-arbit} \\
        \text{s.t.   } & z(E(S)) \leq |S| - 1   & \forall S \subseteq V \nonumber \\
        & z(\delta(v_1)) = 2k & \nonumber \\ 
        & z(E) = n - 1 & \nonumber \\
         &  0 \leq z_e \leq 1 & \forall e \in E \nonumber.
    \end{align}
\begin{claim} \label{claim:tree-unrestricted-arbit}
   Let $T^*_{2k}$ be a minimum cost spanning tree among all spanning trees in $G$ that have degree $2k$ on vertex $v_1$. Then we have that the indicator vector $T^*_{2k}$ is an optimal solution to LP \eqref{LP:tree-unrestricted-arbit}.
\end{claim}
\begin{proof}
    Let $\mathcal{M}_1$ be the graphic matroid on $G$ and $\mathcal{M}_2$ be a partition matroid with one part containing all edges incident to $v_1$ with capacity $2k$ and another part containing the remaining edges with capacity $n-1 -2k$. Then $T^*_{2k}$ is an optimal common base in $\mathcal{M}_1, \mathcal{M}_2$ and the polytope of common bases is give by the constraints of LP \eqref{LP:tree-unrestricted-arbit} .  
\end{proof}

Next we show another LP which will be used for a $\MVmTSP$ problem with depots. For a set of vertices $D $, let $\hat{G}$ be the graph with all vertices of $D$ contracted into a single vertex $\hat{d}$ and for $S \subseteq E - E(D, D)$ let $\kappa_{\hat{G}}(S)$ be the number of components in the graph $\hat{G}$ with edges $S$.
    \begin{align} 
        \text{minimize   } &\sum_{e \in E - E(D, D)}c_ex_e \label{LP:tree-depot} \\
        \text{s.t.   } & x(S) \geq \kappa_{\hat{G}}(\overline{S} - E(D, D)) - 1 & \forall S \subseteq E - E(D, D) \nonumber \\
         & x(E(d, V - D)) \geq 1 & \forall d \in D \nonumber \\
        &0 \leq x_e \leq 1 & \forall e \in E - E(D, D) \nonumber.
    \end{align}
We show the following claim which follows from the ideas from \cite{Cerdeira94}.
\begin{claim}
    A $D$-forest cover is a forest cover such that each component includes exactly one vertex from $D$ and each component has at least two vertices. Then the value of LP \eqref{LP:tree-depot} is the cost of the min cost $D$-forest cover.
\end{claim}
\begin{proof}
    For a $D$-forest cover $F$ we observe that $\overline{F}$ is a edge set that satisfies $ |\{\{v, d\} \in \overline{F}, v \notin D\}| \leq n - |D|$ for all $d \in D$ and that $F$ is spanning tree in the graph $\hat{G}$. We now define matroids $\mathcal{M}_1, \mathcal{M}_2$ on the ground set $E - E(D, D)$ with independent sets $\mathcal{I}_1, \mathcal{I}_2$ whose common independents set correspond to complements of $D$-forest covers. We define $\mathcal{M}_1$ as a partition matroid which has parts $P_d$ for each $d \in D$ that contain edges $\{ \{d, v\} | v \notin D \}$ with capacity $n- |D| -1$ and all other edges $e \notin E(D, D)$ go in a unique part $P_e$ with capacity $1$. We define $\mathcal{M}_2$ as the dual of the graphic matroid on graph $\hat{G}$. Then the complements of $D$-forest covers are common independent sets of $\mathcal{M}_1, \mathcal{M}_2$ and the complements of all common independent sets of $\mathcal{M}_1, \mathcal{M}_2$ contain $D$-forest covers. This implies that the cost of a min cost $D$-forest cover is $c(E - E(D, D)) -  \max_{F \in \mathcal{I}_2 \cap \mathcal{I}_2  }c(F)$. By the characterization of the matroid intersection polytope the value of $\max_{F \in \mathcal{I}_2 \cap \mathcal{I}_2  }c(F)$ is 
    \begin{align}
        & \max \sum_{e \in E - E(D, D)}c_ex_e \nonumber \\
        \text{s.t.   } & \sum_{e \in E(d, V - D)}x_e \leq n - |D| - 1 & \forall d \in D \nonumber \\
        & x(S) \leq r_{\mathcal{M}_2}(S) & \forall S \subseteq E - E(D, D) \nonumber \\
        & 0 \leq x_e \leq 1 & \forall e \in E - E(D, D) \nonumber.
    \end{align}
    Similar to Claim \ref{claim:LP:tree} we make the variable change $z_e = 1 - x_e$, add $\sum_{e \in E - E(D, D)}c_e$, negate the objective, and make it minimization to get LP \eqref{LP:tree-depot}. This follows since the following hold
    \begin{enumerate}
        \item $\sum_{e \in E - E(D, D)}c_e - \sum_{e \in E - E(D, D)}c_ex_e = \sum_{e \in E - E(D, D)}c_ez_e$
        \item $\sum_{e \in E(d, V - D)}x_e \leq n - |D| - 1 \iff \sum_{e \in E(d, V - D)}z_e \geq 1 $ 
        \item $x(S) \leq |S| - \kappa_{\hat{G}}(E - E(D, D) - S) + 1 \iff z(S) \geq \kappa_{\hat{G}}(E - E(D, D) - S) - 1 \iff z(S) \geq \kappa_{\hat{G}}(\overline{S} - E(D, D)) - 1$.
    \end{enumerate}
        Thus $x^*$ is an optimal solution if and only if $z^* = \mathbf{1} - x^*$ is an optimal solution to LP \eqref{LP:tree-depot}.
  \end{proof}

\section{MV-mTSP Arbitrary Tours} \label{Section:arbitrary-depot}
\subsection{Approximation Algorithm with Depots}
In this subsection, we give a $2$-approximation for the $\MVmTSP_+$ with arbitrary tours problem. Here we are given a set of depot vertices $D$ with $|D| = k$ and the goal is to find $k$ closed walks that each closed walk uses exactly one depot. We note that $r(v) = 1$ for all $v \in D$ since we require all walks to be non-empty.  For the single visit case of the problem, a $2$-approximation was given by \cite{many-visit}. This algorithm is a simple combinatorial tree-doubling algorithm. To allow us to use this algorithm in our reduction, we first show that the tree doubling algorithm achieves a $2$-approximation relative to its LP relaxation for the single visit case in Lemma~\ref{lem:guarantee-tree-double}. We use the following LP which comes from LP \eqref{LP:general-single} by setting $D$ to the set of depot vertices 
    \begin{align} 
        \text{minimize   } &\sum_{e \in E}c_ex_e \label{LP:depotSingleVisit} \\
        \text{s.t.   } & x(\delta(v)) = 2   & \forall v \in V \nonumber  \\
        & x(\delta(S \cup D)) \geq 2 & \forall S \subset V - D \nonumber \\
        & x(E(D, D)) = 0 & \nonumber \\
         &0 \leq x_e \leq 2 & \forall e \in E \nonumber.
    \end{align}
    We recall the definition of a $D$-forest cover.
    \begin{definition}
        A $D$-forest cover is a forest cover such that each component includes exactly one vertex from $D$ and each component has at least two vertices.
    \end{definition}
    \begin{algorithm}
    \caption{Tree Doubling Algorithm for $\mtsp_+$ }
\DontPrintSemicolon
\label{Algorithm:mv-TSP-arbitrary}
\KwInput{$G=(V, E), D \subseteq V, c: V \times V \to \mathbb{R}_{\geq 0}$ with $|D| = k$}
\setcounter{AlgoLine}{0}
\KwOutput{$k$ cycles spanning the graph such that each cycle contains exactly one vertex from $D$}

Find a min cost $D$-forest cover .\;

Double all the edges in the $D$-cover and then shortcut so that each vertex is visited exactly once and return the resulting cycles.\;
\end{algorithm}

In the following claim, we show that the cost of min-cost $D$-forest cover as given by LP~\eqref{LP:tree-depot} is at most the cost of the optimal solution to LP~\eqref{LP:depotSingleVisit}.
\begin{claim} \label{claim:tree-doubling-2}
    Let $z^*$ be an optimal solution to LP \eqref{LP:depotSingleVisit} and $x^*$ be an optimal solution to LP \eqref{LP:tree-depot} then we have $c^Tx^* \leq c^Tz^*$.
\end{claim}
\begin{proof}
    Let $z$ be a feasible solution to LP \eqref{LP:depotSingleVisit}, we will show that $z$ is a feasible solution to LP \eqref{LP:tree-depot}. We note that LP \eqref{LP:tree-depot} is defined on edges $E - E(D, D)$ and $z$ is defined on edges $E$, but $z$ satisfies $z_e = 0$ for all $e \in E(D, D)$.  For all $d \in D$ we have that $z(E(d, V - D)) = z(\delta(d)) - z(E(d, D)) =  z(\delta(v)) = 2 > 1$ where the second equality follows since $0 \leq z(E(d, D)) \leq z(E(D, D))  = 0$. We recall that we use the graph $\hat{G}$ in LP \eqref{LP:tree-depot} which we get by contracting all vertices in $D$ to a single vertex. Let $\hat{d}$ be the contracted depot vertex in $\hat{G}$ and for $S \subseteq E - E(D, D)$ let $C_1, \ldots, C_p$ be the components of the sub-graph of $\hat{G}$ with edges $E - S - E(D, D)$ such that $\hat{d} \in C_1$. Then we have that $z(S) \geq  \sum_{i < j \leq p}z(E(C_i, C_j)) =\frac{1}{2} \left(z(\delta((C_1 - \hat{d}) \cup D )) + \sum_{i = 2}^p z(\delta(C_i)) \right) \geq p = \kappa_{\hat{G}}(\overline{S} - E(D, D)) > \kappa_{\hat{G}}(\overline{S} - E(D, D)) - 1$. The second to last inequality follows since for $i > 1$ we have $C_i \subseteq V - D$ so $z(\delta(C_i)) = z(\delta(D \cup C_i')) \geq 2$ for some $C_i' \subseteq V - D$.
\end{proof}

The following lemma now follows straightforwardly about the LP-based guarantee for Tree Doubling Algorithm.

\begin{lemma}\label{lem:guarantee-tree-double}
Let $z^*$ be an optimal solution to linear programming relaxation for the $\mtsp_+$,  LP~\eqref{LP:depotSingleVisit}. then the output of the Tree Doubling Algorithm returns a solution whose cost is at most twice the objective value of $z^*.$
\end{lemma}

We now present the reduction to achieve a $2$-approximation for the $\MVmTSP_+$ with arbitrary tours using Lemma~\ref{lem:guarantee-tree-double} and our general reduction in Theorem~\ref{theorem:general}.

\begin{theorem}
    There is a polynomial time $2$-approximation algorithm for the $\MVmTSP_+$ with arbitrary tours problem.
\end{theorem}
\begin{proof}
The following linear program is a relaxation for the $\MVmTSP_+$ problem.
    \begin{align} 
        \text{minimize   } &\sum_{e \in E}c_ex_e \label{LP:depotMv} \\
        \text{s.t.   } & x(\delta(v)) = 2   & \forall v \in D \nonumber  \\
        & x(\delta(v)) = 2r(v)   & \forall v \in V - D \nonumber  \\
        & x(\delta(S \cup D)) \geq 2 & \forall S \subset V - D \nonumber \\
        & x(E(D, D)) = 0 & \nonumber \\
         &0 \leq x_e \leq 2 & \forall e \in E \nonumber.
    \end{align}
Observe that the LP~\eqref{LP:depotMv} is exactly the same as LP \eqref{LP:general-multi} in the general framework. Moreover, LP~\eqref{LP:depotSingleVisit} is exactly the same as LP~\eqref{LP:general-single}. From Lemma~\ref{lem:guarantee-tree-double}, we obtain that Tree Doubling Algorithm satisfies the condition as needed for Theorem~\ref{theorem:general}. Thus applying Theorem~\ref{theorem:general}, we obtain an integral solution whose cost is at most $2$ times the cost of the optimal solution to LP~\eqref{LP:general-multi} which is at most the cost of the optimal solution to the $\MVmTSP_+$ with arbitrary tours problem as claimed.
\end{proof}
    
\subsection{Approximation for Unrestricted Variant}
In this subsection we show there is a $2$-approximation for unrestricted $\mathrm{mTSP}_+$ with arbitrary tours. We note that we allow using loops for the single visit version here.  Our algorithm is the following.  

\begin{algorithm}[H]
\caption{Unrestricted $\mathrm{mTSP_+}$}
\DontPrintSemicolon
\label{Algorithm:unrestricted-arbit}
\KwInput{$G, k \in \mathbb{Z},1 \leq k \leq n$}
\setcounter{AlgoLine}{0}
\KwOutput{$k$ cycles that cover all vertices in the graph}

Add a new vertex $d$ that has all edges to vertices of $G$ to get a new graph $G'$. Extend the cost function of the graph by setting $c_{dv}  = \frac{c_{vv}}{2}$ .\;

Find a tree of minimum cost that has degree $k$ on the vertex $d$ .\;

Remove the vertex $d$ and all edges incident to it. Among the remaining $k$ components, if the component is a singleton then add a loop in that component. Otherwise double the edges of the tree in the remaining components and shortcut so that each component is a cycle. Return the resulting $k$ cycles .\;
\end{algorithm}
Now we describe our LP relaxation for this problem. We note that our LP is just the LP relaxation of the tree that integrals solutions must contain. For the previous problem this was implicitly true as we showed a fractional solution $x$ was in the up-hull of the tree polytope. In this case the tree is defined on a different graph $G'$ than $x$ is defined on. We denote by $E' = E \cup  \{ \{d, v \} | v \in V\}$ as the edge set of $G'$ where $d$ is a new dummy vertex that we added to the graph. We extend the cost function by setting $c_{dv} = \frac{c_{vv}}{2}$ for all $v \in V$. 

    \begin{align} 
        \text{minimize   } &\sum_{e \in E'}c_ez_e \label{LP:unrestricted-arbit} \\
        \text{s.t.   } & z(E'(S)) \leq |S| - 1   & \forall S \subseteq V \cup \{d\}  \\
        & z(\delta(d)) = k & \nonumber \\ 
        & z(E') = n & \nonumber \\
         &  0 \leq z_e \leq 1 & \forall e \in E' \nonumber
    \end{align}

\begin{claim}
    Let $\mathrm{OPT}$ be the optimal value of a optional solution to the unrestricted mTSP problem with $k$ salesperson and $z^*$ be an optimal solution to LP \eqref{LP:unrestricted-arbit}. Then we have $c^Tz^* \leq \mathrm{OPT}$.
\end{claim}
\begin{proof}
    Let $C_1, C_2, \ldots, C_k$ a feasible solution then we will construct a feasible solution $z$ to LP \eqref{LP:unrestricted-arbit} whose cost (in the extended graph $G'$) is at most the cost of $C_1, \ldots, C_k$. To get $z$, we construct a spanning tree in the graph $G'$ such that the dummy vertex $d$ has degree $k$. To construct the tree $T$, for each $C_i$ we arbitrarily remove one edge from $C_i$ and add an edge between one of the endpoints of the missing edge and the dummy vertex $d$. For each $C_i$, if we remove edge $e = \{u, v\}$ with cost $c_e$ we add either $\{d, v\}$ or $\{d, u\}$ which only decreases the cost since $c_{du} = \frac{c_{uu}}{2} \leq c_e$ and  $c_{dv} = \frac{c_{vv}}{2} \leq c_e$.
\end{proof}

\Umtsp
\begin{proof}
 Let $M_1 ,\ldots, M_k$ be the output of Algorithm \eqref{Algorithm:unrestricted-arbit} and $z^*$ be an optimal solution to LP \eqref{LP:unrestricted-arbit}. We will show that $\sum_{i = 1}^kc(M_i) \leq 2c^Tz^*$. Let $T^*$ be the tree from the second step of Algorithm \eqref{Algorithm:unrestricted-arbit}. First we show that $\sum_{i = 1}^kc(M_i) \leq 2c(T^*)$. If we acquire $M_i$ by adding a loop to to a singleton vertex $v$ then the cost of $M_i$ is $c_{vv}$ while the cost of the edge adjacent to $M_i$ in $T^*$ is $\frac{c_{vv}}{2}$ which is twice the cost in the algorithms output. Now we consider when $M_i$ is a non-singleton component, let $\{d, v\}$ be the edge in $T^*$ such that $v \in M_i$ and let $u \in V$ such that $\{u, v\}$ is an edge in $T^*$. Then we have that $ c_{uv} + \frac{c_{vv}}{2} \leq 2c_{uv}$. For any other edge $e \in M_i$, the algorithm pays at most $2c_e$ while the cost in $T^*$ is $c_e$. Then summing up the cost of all cycles and applying these bounds gives $\sum_{i = 1}^kc(M_i) \leq 2c(T^*)$. Then the proof is concluded by observing $c(T^*) = c^Tz^*$ since $T^*$ is an optimal solution to LP \eqref{LP:unrestricted-arbit} by Claim \ref{claim:tree-unrestricted-arbit}.
\end{proof}

\section{Single Depot Multi-Visit mTSP}
In this section we give a $\frac{3}{2}$-approximation for the $\SDMVTSP_+$ with arbitrary tours problem and a $\frac{7}{2}$-approximation for the $\SDMVTSP_+$ with vertex disjoint tours problem.
\subsection{Approximation Algorithm for Arbitrary Tours}
First we convert the algorithm for $\SDTSP$  by Frieze \cite{Frieze} to a LP analysis since Frieze shows that this algorithm achieves a $3/2$ approximation relative to the integral optimal solution.
\begin{algorithm}[H]
\caption{Single Depot mTSP}
\DontPrintSemicolon
\label{Algorithm:TSP_single}
\KwInput{$G=(V = \{v_1, \ldots, v_n\}), c: V \times V \to \mathbb{R}_{\geq 0}, k \in \mathbb{N}$}
\setcounter{AlgoLine}{0}
\KwOutput{$k$ cycles that contain $v_1$ that span the graph and vertices not equal to $v_1$ are visited exactly once}

Among spanning trees $T$ such that $d_T(v_1) = 2k$ find a min cost tree $T^*$ .\;

Find a min cost perfect matching $M^*$ on the vertices with odd degree in $T^*$.\;

Add $M^*$ to $T^*$ which is now a Eulerian graph. Let $w_1 = v_1, \ldots,w_s = v_1$ be the Eulerian tour and let $U$ be the neighbors of $v_1$ in $T^*$. Delete a node $w_i$ in the sequence if
\begin{enumerate}
    \item $w_i$ has appeared before and $w_i \neq v_1$ or
    \item $w_i \in U$ and $v_1 \notin \{ w_{i-1}, w_{i + 1} \}$.\;
\end{enumerate} 

Return the sequence obtained after short cutting.\;
\end{algorithm}
We use the following LP for $\SDTSP$. We note that the LP does not exactly fit the LP in the general framework (LP \eqref{LP:general-single}) since there is a different constraint on the degree of vertex $v_1$. We still use the  general framework in this section, but we show that each part of the framework still holds with the additional degree constraint. 
\begin{empheq}{align} \label{LP: 1DTSP}
\text{minimize   } &\sum_{e \in E}c_ex_e \\
\text{s.t.   } & x(\delta(v)) = 2 &\forall v \in V - v_1 \nonumber \\
& x(\delta(v_1)) = 2k \nonumber \\
 & x(\delta(S)) \geq 2 &\forall S \subseteq V \nonumber \\
 & x_e \geq 0 & \forall e \in E. \nonumber
\end{empheq}

We first show the cost of $T^*$ is at most the cost of the LP relaxation for the problem. 
\begin{lemma}
    Let $x^*$ be an optimal solution to LP \eqref{LP: 1DTSP}. Then we have, 
    \[ c(T^*) \leq c^Tx^*. \]
\end{lemma}
\begin{proof}
We show that any solution $x$ to LP \eqref{LP: 1DTSP} is feasible for LP \eqref{LP:tree} which will conclude the proof. The solution $x$ clearly satisfies $x(\delta(v_1)) = 2k$ and $0 \leq x$ by the constraints of LP \eqref{LP: 1DTSP}. Let $S \subset E$ and $C_1, \ldots, C_m$ be a partition of the vertex set $V$ in the graph $(V, \overline{S})$. Then we have that,
\begin{align*}
    x(S) &\geq \sum_{i < j}x(E(C_i, C_j)) = \frac{1}{2} \sum_{i = 1}^k x( \delta(C_i) ) \\
    &\geq m = \kappa(\overline{S}) >  \kappa(\overline{S}) - 1.
\end{align*}
The first inequality follows since $E(C_i, C_j) \subseteq S$ since $C_1, \ldots, C_k$ are components in the graph with edges $\overline{S}$ and the second inequality follows since $x$ is feasible for LP \eqref{LP: 1DTSP}.
\end{proof}

We can show the cost of the matching is a at most $1/2$ the cost of the LP optimum. 
\begin{claim}
    Let $x^*$ be an optimal solution to LP \eqref{LP: 1DTSP}. Then we have $c(M^*) \leq \frac{c^Tx^*}{2}$.
\end{claim}
\begin{proof}
    Let $S$ be the set of odd degree vertices in $T*$ then $M^*$ is a min-cost $S$-join in $G$. The polytope for $S$-joins is given by $\{x \geq 0 | x(\delta(P)) \geq 1, \forall P \text{ such that } |P \cap S| \text{ is odd} \}$. Then the claim follows since $x^*/2$ is a feasible solution for the $S$-join polytope.
\end{proof}

Then the above two lemmas imply the following.
\begin{lemma} \label{lemma:friezeLP}
    Let $x^*$ be an optimal solution to LP \eqref{LP: 1DTSP} Algorithm \eqref{Algorithm:TSP_single} returns a solution $C$ satisfying $\sum_{e \in C}c_e \leq \frac{3}{2}c^Tx^*$.
\end{lemma}
\begin{proof}
    This follows since by the triangle inequality $\sum_{e \in C}c_e \leq c(M^*) + c(T^*) \leq \frac{3}{2}c^Tx^*$.
\end{proof}

Now we are ready to get a $3/2$ algorithm for the multi-visit variant. We need the following to characterize solutions to the problem.

\begin{lemma}
    Given a connected graph with edge set $T$ such that $d_T(v_1) = 2k$ and $d_T(v) = 2r(v)$, we can can decompose the edges of $T$ into $k$ closed walks containing $v_1$.
\end{lemma}

\begin{proof}
    The graph $G$ is Eulerian so there exists an Eulerian $C$ circuit starting at $v_1$ and the circuit is given by a sequence of vertices $w_1 = v_1, \ldots, w_k = v_1$. Let $w_1, w_2, \ldots, w_j$ be a prefix of the sequence such that $j$ is the smallest index greater than $1$ such that $w_j = v_1$. We will use $w_1, \ldots, w_j$ as the first closed walk. Next we reduce the graph by deleting all edges used by the first closed walk and then by removing any isolated vertices. We now show this graph is still Eulerian. Clearly all vertices have even degree since we removed an even number of edges from each vertex. The graph remains connected since the Eulerian circuit $C$ will not use any of the removed vertices or edges in the graph. Thus we can inductively repeat this process to get $k$ closed walks containing $v_1$ so that each vertex $v$ is visited a total of $r(v)$ times. 
\end{proof}

We use following LP for the multi-visit version of the problem. 
    \begin{align} 
        \text{minimize   } &\sum_{e \in E}c_ex_e \label{LP:oneDepot} \\
        \text{s.t.   } & x(\delta(v_1)) = 2k   & \nonumber  \\
        & x(\delta(v)) = 2r(v) & \forall v \in V \nonumber \\
        & x(\delta(S)) \geq 2 & \forall S \subseteq V \nonumber \\
         & x_e \geq 0 & \forall e \in E \nonumber.
    \end{align}

The proof of this is nearly identical to Lemma \ref{lemma:general-smallR}.
\begin{claim} \label{claim:smallR-1d}
    Let $\mathcal{A}$ be an algorithm that takes a solution $z$ to LP \ref{LP:oneDepot} and outputs $k$ cycles $t_1, \ldots, t_k$ satisfying $\sum_{i = 1}^kc(t_i) \leq \rho c^Tz$. Then given a solution $x$ to LP \eqref{LP: 1DTSP} there is an algorithm that outputs a feasible solution to $\SDMVTSP$, $T: E \to \mathbb{Z}$ satisfying $\sum_{e \in E}c_e T(e) \leq \rho c^Tx$ in time polynomial in $\max_{v \neq v_1}r(v)$ and $n$.
\end{claim}
\begin{proof}
    We follow the proof of Lemma \ref{lemma:general-smallR}. We construct the graph $G^r$ identically as in Lemma \ref{lemma:general-smallR} for the non-depot vertices meaning $G^r$ has $r(v)$ copies of each vertex $v \in V - v_1$ and for $v_1$ the graph $G^r$ has only one copy. We also extend the cost function $c$ where an edge $\{u_i, v_j\}$ has cost $c_e$ where $e = \{u, v\}$ is the edge with between the original vertices that $u_i, v_j$ are copies of. For each edge $e = \{u_i, v_j\} \in E^r$ where $u_i, v_j$ are copies of $u, v \in V$ we construct a solution $z_e = \frac{x_{uv}}{r(u)r(v)}$ where in this context we set $r(v_1) = 1$. 
    
    We show that $z$ is a feasible solution to LP \ref{LP:oneDepot} for graph $G^r$ and satisfies $c^Tx = c^Tz$ which follows from Lemma \ref{lemma:solution-construction} with one slight addition. We use Lemma \ref{lemma:solution-construction} by setting $D = \{v_1\} $ and observe that the graph $G^r$ and cost function $c$ in the Lemma \ref{lemma:solution-construction} are the same as $G^r, c$ given in this lemma. Then in Lemma \ref{lemma:solution-construction} we showed that $z(\delta(v')) = \frac{x(\delta(v))}{r(v)}$ where $v' \in V^r$ is a copy of vertex $v \in V$ which implies for $v \neq v_1$ that $z(\delta(v)) = 2$ and $z(\delta(v_1)) = 2k$. Finally in Lemma \ref{lemma:solution-construction} we showed that $z(\delta(S \cup d)) \geq 2$ for all $S \subset V^r - \{d\}$ which is equivalent to $z(\delta(T)) \geq 2$ for all $T \subset V^r$ since if $d \notin T$ then we have that $d \in V^r - T$ implying $z(\delta(T)) = z(\delta(V^r - T)) \geq 2$. 

    Thus by applying algorithm $\mathcal{A}$ to $G^r$ we get $k$ cycles containing $v_1$ that visit all vertices $v \neq v_1$ once. These $k$ cycles correspond to $k$ closed walks $t_1, \ldots, t_k$ in $G$ starting at $v_1$ that together satisfy the visit requirements for all vertices. We also have that $\sum_{i = 1}^kc(t_i) \leq \rho c^Tx$ since the closed walks in $c_1, \ldots, c_k$ have the same cost as the cycles in $G^r$ and the cycles in $G^r$ have cost at most $c^Tz = c^Tx$. The run-time follows since the maximum number of vertices in $G^r$ is $n \max_{v \neq v_1}r(v)$.
\end{proof}

The following algorithm is nearly identical to Algorithm \eqref{Algorithm:general} when we set $D = \{v_1\}$.
\begin{algorithm}[H]
\caption{Multi-Visit TSP Single Depot Arbitrary}
\DontPrintSemicolon
\label{Algorithm:TSP_single-mv}
\KwInput{$G=(V = \{v_1, \ldots, v_n\}), c: V \times V \to \mathbb{R}_{\geq 0}, m \in \mathbb{N}, r: V - v_1 \to \mathbb{Z}$}
\setcounter{AlgoLine}{0}
\KwOutput{$k$ tours that contain $v_1$ such that vertices $v \neq v_1$  are visited $r(v)$ times}

Solve LP \eqref{LP:oneDepot} to get solution $x^*$ .\;

For all edges $e$ let $\tilde{x}_e = x_e - 2k_e$ such that $k_e = 0$ if $x_e \leq 4$ and otherwise $k_e$ is set so that $2 \leq \tilde{x}_e < 4$ and $k_e \in \mathbb{Z}$. Define a function $\tilde{r}: V - v_1 \to \mathbb{Z}$ where $\tilde{r}(v) = r(v) - \sum_{e \in \delta(v)}k_e$ and $\tilde{k} = \frac{1}{2}\tilde{x}(\delta(v_1))$.\;

Use Claim \ref{claim:smallR-1d} with solution $\tilde{x}$ on instance $G, \tilde{r}, \tilde{k}$ .\;

Increase the number of times each edge is used in the previous step by $2k_e$ and return the resulting solution.\;
\end{algorithm}

As in the previous section we show the following lemmas and claims to show this algorithm gets a $\rho$ approximation.

The proof of this claim is identical to Claim \ref{claim:rbound}.
\begin{claim}
    For all $v \in V - v_1$ we have $1 \leq \tilde{r}(v) \leq 2n$ and $\tilde{k} \geq 1$.
\end{claim}
\begin{proof}
    The proof of Claim \ref{claim:rbound} shows $1 \leq \tilde{r}(v) \leq 2n$ for all $v \in V - v_1$. Now we show $\tilde{k} \geq 1$.  If we have that if $x_e \leq 4$ for all $e \in \delta(v_1)$ then $\tilde{k} = k \geq 1$ otherwise if there exists $e$ such that $x_e > 4$ then $\tilde{k} = \frac{1}{2}\tilde{x}(\delta(v_1)) \geq \frac{\tilde{x}_e}{2} > 2$.
\end{proof}

\begin{claim}
    The solution $\tilde{x}$ is a feasible solution for LP \eqref{LP:oneDepot} with graph $G$ and $\tilde{r},\tilde{k}$.
\end{claim}
\begin{proof}
    By definition we have $\tilde{x}(\delta(v_1)) = 2\tilde{k}$ and the rest claim follows from the proof of Claim \ref{claim:general-feasible}.
\end{proof}

Then we get the following theorem whose proof is identical to the proof of Theorem \ref{theorem:general}.  
\begin{theorem}
    Let $x^*$ be an optimal solution to LP \eqref{LP:oneDepot} and $\rho$ be the approximation factor of algorithm $\mathcal{A}$ for $\SDTSP$ whose guarantee is relative to the value of LP \eqref{LP: 1DTSP}. Then Algorithm \eqref{Algorithm:TSP_single-mv} outputs a solution to $\SDMVTSP$, $T: E \to \mathbb{Z}$ satisfying $\sum_{e\in E}T(e) \leq \rho c^Tx^*$ and runs in time polynomial in $n$. 
\end{theorem}
This gives the following corollary which we get by using the analysis of the Frieze algorithm we showed at the beginning of the section. 
\begin{corollary}
    There is an approximation algorithm for the $\SDMVTSP$ problem with an approximation factor of $\frac{3}{2}$.
\end{corollary}

\subsection{Approximation Algorithm for for Vertex Disjoint Tours}
Here we show an algorithm for the vertex disjoint variant that achieves a $7/2$-approximation. We note that this result does not follow the general framework and follows from a simple use of the single visit algorithm. The idea for this algorithm is from \cite{many-visit} which is that we can first find a $\mtsp$ solution that visits all vertices once. Then we add loops to the different tours to satisfy the visit requirements and adding loops allows us to maintain the vertex disjoint property that all solutions to single visit variants have. 

\begin{algorithm}[H]
\caption{Multi-Visit TSP Single Depot Vertex Disjoint}
\DontPrintSemicolon
\label{Algorithm:TSP_single-mv-dis}
\KwInput{$G=(V = \{v_1, \ldots, v_n\}), c: V \times V \to \mathbb{R}_{\geq 0}, k \in \mathbb{N}, r: V - v_1 \to \mathbb{Z}$}
\setcounter{AlgoLine}{0}
\KwOutput{$k$ closed walks that contain $v_1$ such that vertices $v \neq v_1$  are visited $r(v)$ times, closed walks are disjoint outside of $v_1$}

Use Algorithm \eqref{Algorithm:TSP_single} to get $k$-cycles containing $v_1$ so that all vertices are visited once.\;

Add $r(v) - 1$ loops to all vertices $v \neq v_1$ to the solution from the previous step.\;

\end{algorithm}

We need the following claim.
\begin{claim} \label{claim:7/2-2}
    Let $\mathrm{OPT}$ be the value of the optimum solution. Then we have $\sum_{v \in V - v_1}r(v) c(vv) \leq 2 \mathrm{OPT}$.
\end{claim}
\begin{proof}
    Let $T: E \to \mathbb{Z}$ be an optimal solution. Then we have 
    \begin{align*}
        \mathrm{OPT} &= \frac{1}{2}\sum_{e \in \delta(v_1)}c_eT(e) + \frac{1}{2}\sum_{v \in V - v_1}\sum_{e \in \delta(v)}T(e)c_e \\
        &\geq \frac{1}{2}\sum_{e \in \delta(v_1)}c_eT(e) + \frac{1}{2}\sum_{v \in V - v_1}2r(v)\min_{e \in \delta(v)}c_e \\ 
        &\geq \frac{1}{2}\sum_{e \in \delta(v_1)}c_eT(e) + \frac{1}{2}\sum_{v \in V - v_1}r(v) c_{vv} \\ 
        &\geq \frac{1}{2}\sum_{v \in V - v_1}r(v) c_{vv}.
    \end{align*}
    The first inequality follows since for all $v \in V -v_1$ we have $\sum_{e \in \delta(v_1)}T(e) = 2r(v)$ and the second inequality follows by the triangle inequality since for any edge $c \in \delta(v)$ we have $c_{vv} \leq 2c_e$.
\end{proof}

Then showing the following claim will imply that we get a $7/2$-approximation. 
\begin{claim} \label{claim:7/2-1}
    Let $c_1, \ldots, c_k$ be the $k$ cycles returned in the first step of the algorithm. Then we have that $\sum_{i = 1}^kc(c_i) \leq \frac{3}{2} \mathrm{OPT}$.
\end{claim}
\begin{proof}
    Let $p_1, \ldots, p_m$ be a optimal solution with value $\mathrm{OPT}$. We have that any two cycles $p_i, p_j$ only intersect at the depot vertex $v_1$ since we are in the vertex disjoint tours setting. For each $p_i$ we can shortcut to get cycle $r_i$, so that all vertices in $p_i$ are visited once and by the triangle inequality we have that $c(r_i) \leq c(p_i)$ implying $\sum_{i = 1}^kc(r_i) \leq \mathrm{OPT}$. Thus we get that $\sum_{i = 1}^kc(c_i) \leq \frac{3}{2} \sum_{i = 1}^kc(r_i) \leq \mathrm{OPT}$ where the first inequality follows since Algorithm \eqref{Algorithm:TSP_single} is a $3/2$-approximation.
\end{proof}
Thus the above two claims imply the following theorem.
\depotDisjoint
\begin{proof}
    Let $c_1, \ldots, c_k$ be the cycle acquired in the first step of the algorithm. The cycles $c_1, \ldots, c_k$ only intersect at the depot vertex $v_1$ since they are a feasible solution to $\SDTSP$. Adding the loops to the cycles keeps this property so Algorithm \ref{Algorithm:TSP_single-mv-dis} outputs a feasible solution. Finally, the cost of the solution is $\sum_{i = 1}^kc(c_i) + \sum_{v \in V-v_1}(r(v) - 1)c_{vv} \leq \frac{7}{2}\mathrm{OPT}$ where the last inequality follows by Claim \ref{claim:7/2-1} and Claim \ref{claim:7/2-2}. The run-time follows immediately since both steps of the algorithm are polynomial time.
\end{proof}

\section{Further Directions}
In this paper we gave a reduction from various multi-visit TSP problems and their respective single visit versions. Our reduction relies on the connection between the LP relaxations of multi-visit variants and their respective single visit variants. There are two open questions that follow naturally.

\textbf{Get a $3/2$ approximation for $\MVmTSP_0$ with arbitrary tours.} For the $\MVmTSP_0$ with arbitrary tours problem, we are given $k$ depots and the visit function $r$ and the goal is to find at most $k$ closed walks so that all non-depot vertices $v$ are visited $r(v)$ times and each closed walk contains exactly one depot. Very recently Deppert, Kaul, and Mnich \cite{VarDepot-3/2} showed the following LP for $\mtsp_0$ has an integrality gap of $2$ and gave a $3/2$-approximation for $\mtsp_0$
    \begin{align} 
        \text{minimize   } &\sum_{e \in E}c_ex_e \label{LP:depotSingleVisit-atmost} \\
        \text{s.t.   }  & x(\delta(v)) = 2 & \forall v \in V - D \nonumber \\ 
        & x(\delta(S \cup D)) \geq 2 & \forall S \subset V - D \nonumber \\
        & x(E(D, D)) = 0 & \nonumber \\
         &0 \leq x_e \leq 2 & \forall e \in E \nonumber.
    \end{align}
    This means we cannot apply our reduction to $\MVmTSP_0$ with arbitrary tours by using LP \eqref{LP:depotSingleVisit-atmost} for $\mtsp_0$. One direction is to get a reduction from  $\MVmTSP_0$ with arbitrary tours to $\mtsp_0$ that does not use LPs.
    
\textbf{Apply the reduction to the unrestricted MV-$\mathrm{mTSP_+}$ with arbitrary tours.} In Section \ref{Section:arbitrary-depot} we gave a $2$-approximation for the $\mathrm{mTSP_+}$ problem and the approximation factor was with respect to the value of LP \eqref{LP:unrestricted-arbit}. We recall that LP \eqref{LP:unrestricted-arbit} was not a LP relaxation where the characteristic vectors of integral solutions to the problem are feasible, but instead it was the convex hull of certain trees that all integral solutions contain. We are not able to apply the reduction described in Section \ref{Section:framework} as the LP does not follow the structure of the LP described in the general framework. In particular it is difficult to find a feasible solution $\tilde{x}$ for the reduced visit function $\tilde{r}$. Either finding a different LP relaxation or finding a different way to apply the reduction to LP \eqref{LP:unrestricted-arbit} would improve the approximation factor of unrestricted $\MVmTSP_+$ from $4$ to $2$.

\bibliography{main}
\bibliographystyle{unsrt}  

\appendix

\section{Missing Proofs From Section \ref{Section:Overview}} \label{Section:TSP-proofs}
\begin{proof}[Proof of Lemma \ref{lemma:smallR-tsp}]
    Let $G^r$ be a complete graph that has $r(v)$ copies of each vertex in the graph with edge set $E^r$ and vertex set $V^r$. We extend the cost function $c$ to $G^r$ and assign edge $\{u_i, v_j\} \in E^r$ cost $c_e$ where  $\{u_i, v_j\}$ is copy of edge $e = \{u, v\}$. We will convert the solution $y$ to a solution $x$ to LP \eqref{LP:tsp} on the graph $G^r$. Let $e' = \{u_i, v_j\} \in E^r$ and $e = \{u, v\}$ are the original copies of $u_i, v_j$ in $V$, then we set $x_{e} = \frac{y_e}{r(u)r(v)}$. For all $v \in V$. Now we show that $x$ is a feasible solution to LP \eqref{LP:tsp} for the graph $G^r$. 
    For any vertex $v \in V^r$ we have $x(\delta(v_i)) = \frac{y(\delta(v))}{r(v)} = 2$ where the first equality follows since the degree of each vertex $v$ in $y$ is distributed evenly among all $r(v)$ copies in $x$ the second equality follows by the feasibility of $y$. 
    
    For any $S \subset  V^r$, we need to show $x(\delta(S)) \geq 2$. Let $k$ be the number of vertices $v \in V$ such that there exists $v_i, v_j$ that are distinct copies of $v$ and $S$ contains exactly one of $v_i, v_j$. We show  $x(\delta(S)) \geq 2$ by induction on $k$. If $k = 0$, then $x(\delta(S)) = y(\delta(S'))$ where $S' \subseteq V$ is acquired by taking the original copy of each vertex $v_j$ from $S$ implying $x(\delta(S)) \geq  2$ since $y(\delta(S')) \geq 2$ since $y$ is feasible to LP \eqref{LP:tsp}. If $k > 0$, then there exists a vertex $v \in V$ such that both $S$ and $V^r - S$ have copies of $v$. Let $B = V^r - S$ $S(v) = S \cap \{v_1, \ldots, v_{r(v)}\}$ and $B(v) =  \{v_1, \ldots, v_{r(v)}\} - S(v)$. For some $v_i \in S(v)$ let $X_1 = x(E(v_i, S - S(v)))$ and $X_2 = x(E(v_i, V^r - S))$. Then we observe that $x(E(S(v) , S - S(v))) = X_1|S(v)|$, $x(E(S(v), B)) = X_2 |S(v)|$ which follows since all copies of vertex $v$ have the same set of neighbors and the same $x$ values assigned to their edges. WLOG, we may assume $X_1 \leq X_2$ otherwise replace $S$ with $B$ since $x(\delta(S)) = x(\delta(B))$. Then we have $x(\delta(S - S(v)))  = x(\delta(S)) + x(E(S(v),S - S(v) )) - x(E(S(v),B))   = x(\delta(S)) + |S(v)|(X_1 - X_2) \leq  x(\delta(S))$ and by induction we get $x(\delta(S - S(v)) \geq 2$ since $S - S(v)$ is a set with one less vertex $v$ that separates a pair of copies of $v$.

    We now show that $0 \leq x \leq 1$. Suppose $\exists e  = \{u, v\} \in E$ with $x_e > 1$. This implies that $x(E(u, V - \{u, v\})) < 1$ and $x(E(v, V - \{u, v\})) < 1$ since $2 = x(\delta(v)) = x(E(v, V - \{u, v\})) + x_e >  x(E(v, V - \{u, v\})) + 1$ and $2 = x(\delta(u)) = x(E(u, V - \{u, v\})) + x_e >  x(E(v, V - \{u, v\})) + 1$. Thus we get that $x(\delta(\{u, v\})) =  x(E(v, V - \{u, v\})) +  x(E(u, V - \{u, v\})) < 2$ which is a contradiction.
    Thus we can apply the algorithm from the lemma assumption on $x$ as a solution to $G^r$ to get a Hamiltonian cycle $C$ in $G'$ satisfying $\sum_{e \in C}c_e \leq \rho c^Tx$. We note that $c^Ty = c^Tx$ since $c^Tx = \sum_{e = \{u, v\} \in E}r(u)r(v)x_e c_e = \sum_{e \in E}c_ey_e $. We now convert $C$ to a closed walk in $G$ by replacing every copy edge $\{u_i, v_j\}$ with its corresponding original edge $\{u, v\}$ in $G$. Clearly, $T$ is a closed walk in $G$. Moreover, $T$ visits every vertex $r(v)$ times $C$ visits each copy of $v$ one time. Finally we the run-time follows since $G^r$ has at most $n \max_{v \in V}r(v)$ vertices and the algorithm in the lemma statement is a polynomial time algorithm. 
\end{proof}

\begin{proof}[Proof of Theorem \ref{theorem:mvMtsp}]
    Let $T'$ be the solution we get from the third step before we increase each edge by $2k_e$. By Claim \ref{claim:rbound-tsp} and Lemma \ref{lemma:smallR-tsp} we have that $\tilde{x}$ satisfies 
    $ \sum_{e \in E} T'(e)c_e \leq  \rho \sum_{e \in E}c_e \tilde{x}_e $. 
    Let $S = \{e \in E | k_e > 0  \}$ be the set of edges then we have that $\sum_{e \in E}T(e)c_e = \sum_{e \in E}T'(e)c_e  + 2\sum_{e \in S} k_e c_e \leq  \rho \sum_{e \in E}c_e \tilde{x}_e  + 2\sum_{e \in S} k_e c_e \leq \rho c^Tx^*$. For the run-time, we get have that Step $3$ runs in time polynomial in $n$ by  Claim \ref{claim:rbound} and Lemma \ref{lemma:general-smallR} and the rest of the steps are clearly polynomial time in $n$. Finally, $T$ is $T$ is a feasible integral solution to LP \eqref{LP:tsp-multi} since $T'$ satisfies the cut constraints which implies $T$ also satisfies the cut constraints since $T(e) \geq T'(e)$ for all $e$ and by definition $T$ satisfies degree constraints. 
\end{proof}

\end{document}